\documentclass[11pt]{article}
\usepackage{fullpage}

\usepackage[round]{natbib}
\usepackage[T1]{fontenc}
\usepackage{lmodern}
\usepackage[utf8]{inputenc} % allow utf-8 input
\usepackage{nameref}
\usepackage[colorlinks, linkcolor=blue, filecolor = blue, citecolor = blue, urlcolor = blue]{hyperref}
\usepackage{url}            % simple URL typesetting
\usepackage{booktabs}       % professional-quality tables
\usepackage{amsfonts}       % blackboard math symbols
\usepackage{nicefrac}       % compact symbols for 1/2, etc.
\usepackage{microtype}      % microtypography
\usepackage{xcolor}         % colors
\usepackage[T1]{fontenc}
\usepackage{lmodern}
\usepackage[utf8]{inputenc} % allow utf-8 input
\usepackage{booktabs}       % professional-quality tables
\usepackage{microtype}      % microtypography

\usepackage{amsmath, amsthm, amssymb, amsfonts}
\usepackage{bm}
\usepackage{mathtools}
\usepackage{tabularx}
\usepackage{thm-restate}
\usepackage{color}
\usepackage{enumitem}

\usepackage{amsthm}

\usepackage{graphics}
\usepackage{algorithm}
\usepackage{caption}
\usepackage{subcaption}
\usepackage{algorithmic}
\usepackage{graphicx} 
\usepackage{adjustbox}
\usepackage{textcomp}
\usepackage{booktabs}
\usepackage{multirow}
\usepackage{multicol}
\usepackage{xcolor}
\usepackage{float}
\usepackage{comment}

\date{}

\newtheorem{definition}{Definition}
\newtheorem{proposition}{Proposition}
\newcommand{\retrieve}{\textsc{RETRIEVE}}
\newcommand{\RomanNumeralCaps}[1]
    {\MakeUppercase{\romannumeral #1}}

\def \papertitle{A Sketch Based Game Theoretic Approach to Detect Anomalous Dense Sub-Communities in Large Data Streams}

\newtheorem{rem}{Remark}

\DeclareMathAlphabet{\pazocal}{OMS}{zplm}{m}{n}

\newcommand{\mrmd}{\textsc{TGDC}}
\def\BibTeX{{\rm B\kern-.05em{\sc i\kern-.025em b}\kern-.08em
    T\kern-.1667em\lower.7ex\hbox{E}\kern-.125emX}}

\title{\papertitle}%\\in the Plackett-Luce model} 
\author{
Prateek Chanda%
\thanks{Microsoft Research, Bangalore, India; {\tt v-pratec@microsoft.com}.} 
\and
Aadirupa Saha%
\thanks{Microsoft Research, New York City, USA; {\tt aadirupa.saha@microsoft.com}.} 
}

\begin{document}

\maketitle

\begin{abstract}
Detecting anomalous subgraphs in a dynamic graph in an online or streaming fashion is an important requirement in industrial settings for intrusion detection or denial of service attacks. While only detecting anomalousness in the system by edge frequencies is an optimal approach, many latent information can get unnoticed in the process, since as a characteristic of the network only edge frequencies are considered. We propose a game theoretic approach whereby using the modularity function we try to estimate in a streaming graph \emph{whether addition of a new edge in the current time tick results in a dense subgraph creation, thus indicating possible anomalous score}. Our contributions are as follows: (a) We propose a novel game-theoretic framework for detecting dense subcommunities in an online streaming environment; (b) We detect such subcommunities using constant memory storage. Our results are corroborated with empirical evaluation on real datasets. 
%\red{1. Intro (Prob Statement + Contribution + Rel Work), 2. Problem Setting (Defs + Motivation of modularity), 3. Proposed method [2 subsections: GT subgraph detections and Myopic 1-hop detection ], 4. Theoretical Guarantees - Estimated anomaly score bounded 5. Expts 6. Contribution  }
\end{abstract}

\vspace{-10pt}
\section{Introduction}
\vspace{-10pt}
Studying the behaviour of agents and detecting anomalies in streaming graphs have gained much importance in recent years due to rising denial of distributed attacks, financial frauds and intrusion in network services. Here we try to solve the problem for detecting dense subgraphs or communities formed in a dynamic network where the network keeps on evolving over time via addition of new nodes or edges.

As graph structures have become more dynamic in nature, particularly in streaming environments where new edges arrive at each time stamp in an online setting, it is crucial to have a constant memory solution for quick detection of intrusion and faster recovery. Recent works like \cite{bhatia2020midas}, \cite{eswaran2018sedanspot} and \cite{eswaran2018spotlight} propose real-time online anomaly detection in large scale networks. Most of these recent anomaly detection models consider mainly edge frequencies (or edge weights as a deciding factor) for anomaly in the network, which often fail to capture the notion of suspicious subgraphs in the system whose corresponding network edges might have less frequency. %\red{cite refs to justify your claims}
%\red{Prob statement (informal?)}

%\red{Main questions we raised and why its important (w/ applications)}
\vspace*{-10pt}
\begin{figure}[H]
    \centering
    \includegraphics[width=0.7\linewidth]{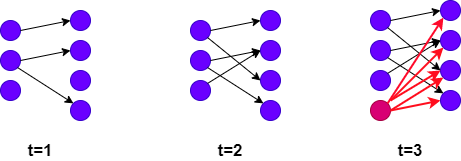}
    \vspace*{-4pt}
    \caption{Sudden appearance of dense subgraphs}
    \label{fig:densesubgraph motivation}
\end{figure}%
\vspace*{-15pt}
Here, we propose a dense subgraph based anomaly detection in online network streams that considers modularity of sub-communities in the network. In addition, we propose a myopic 1-hop based game-theoretic approach that minimizes penalties when forming new communities. 
Our precise contributions are as follows:
\begin{itemize}[noitemsep,topsep=0pt]
    \item \textbf{Modularity Based Dense Subgraph Detection}: In Sections \ref{Proposed Methodology} \& \ref{anomaly indication} we present our proposed Game-Theoretic Framework with long term Reward Maximization for finding dense subgraphs in real time. %\red{put ref. to the section numbers} \blue{resolved}
    \item \textbf{Constant Space}: Using constant space for our model to approximate dense subgraphs within the network.
    
    \item \textbf{Theoretical Guarantees}: We showcase that the using our proposed model in subsequent theorems (see Thm. \ref{proofs}), Nash equilibrium is preserved.
    
    \item \textbf{Real World Effectiveness}: Our experimental results in Section \ref{experiments} show how our proposed model is able to detect dense anomalous subcommunities. %\red{put ref. to the section numbers} \blue{resolved}
\end{itemize}

\subsection{Related Workss}
%\red{Rel. works. }

Recent methods focusing on anomaly detection models in a dynamic network setting can be broadly categorized into the following:

\begin{enumerate}
    \item \textbf{Distribution Learning Based Approach}
    Distribution shifts or sudden changes/spikes %\red{cite} 
    in edge frequency distribution is a common measure for anomaly detection. MIDAS \cite{bhatia2020midas} and its extension MIDAS-F \cite{bhatia2020real} uses chi-square test statistic \cite{zibran2007chi} to compute whether an incoming edge is an anomaly based upon the current edge frequency and the mean frequency (calculated via previous frequency distribution for the same edge). Other existing methods like \cite{eswaran2018sedanspot} and \cite{eswaran2018spotlight}
    \cite{}

    \item \textbf{Dense Subgraph as an Indication of Anomaly}
    \cite{hooi2020telltail} uses Extreme Value Theory and probabilistic measures to evaluate the sudden appearance of dense subgraphs in timeseries networks, but however fails to capture for large data storage where constant memory analysis is required. Other recent methods that try to find dense subgraphs consider edge surplus as a measure of how closely nodes within a subgraph are interconnected \cite{tsourakakis2013denser}
or through k-clique discovery  %\red{cite??}.
%\red{State what these methods have done and what are the limitations.}

\item \textbf{Reinfrocement Learning based dense subgraphs discoveries.}
Recent methods, e.g. \cite{pang2021toward},  used a semi-supervised approach to detect anomalous data via deep reinforcement learning, while \cite{wu2021rlad} uses the same setting coupled with active learning to detect anomalous data points in timeseries. %\red{State what these methods have done and what are the limitations.}

\end{enumerate}

%\red{Optional: Roadmap}
\vspace{-8pt}
\section{Preliminaries and Problem Setting}
\vspace{-7pt}

We start by introducing some notations as listed below.

\begin{table}[h!]
\begin{tabular}{ll}\toprule
\parbox[t]{5cm}{\textbf{Symbols}} &  \parbox[t]{5cm}{\textbf{Definition}} \\ \toprule
A   & Real symmetric $n $ x $n$ adjacency matrix representing the network. \\
\midrule
$\delta$ & Kronecker delta function. 

$\delta_{ij}$ = 
$\begin{cases}
      1, i = j \\
      0, i \neq j
    \end{cases}$
\\
\midrule
Q & Modularity Function \\
\midrule
$\gamma$ & Resolution Parameter. \\
\midrule
$q(u,v,C)$ & The payoff function of node $u$ when node $v$ joins its community $C$\\
\midrule
$\phi(u)$ & Community label assigned to node $u$\\
\midrule
$m$ &  Total edges in the network \\
\midrule
$k_i$ & Degree of node $i$\\
\midrule
$\mu(i,C)$ & Utility of node $i$ to join community $C$ \\
\midrule
$C_x$ & Community labelled as $x$ \\
\midrule
$\mu(C_x,C_y)$ & Utility of $C_x$ when it wants to merge with $C_y$\\

\bottomrule 
\end{tabular}
\end{table}

\vspace{4pt}
\textbf{Problem Setting. } Given a dynamic time-evolving graph $G$ with the set of incoming edges denoted by $\varepsilon = (e_1, e_2, \dots)$ and where $e_i:=(u_i, v_i, t_i)$ denotes the incoming edge between nodes $u_i$ and $v_i$ at time instant $t_i$, how can we detect anomalous dense subgraphs. For instance, in  network logs $u_i$ and $v_i$ may indicate source and destination IP addresses respectively and their corresponding edge denoting a connection request occurring at timestamp $t_i$. We also consider the case where there might be repetitions of the same incoming edge in future i.e. the corresponding network $G$ of incoming edges can be modelled as a directed multi-graph where $(u_i,v_i) \neq (v_i, u_i)$.
We try to answer the following questions through our proposed framework :

\begin{itemize}[noitemsep,topsep=0pt]
    \item Uptil time step $t$, can we observe such closely connected subgraphs and make a decision when a new edge arrives, due to it does the subgraph modularity increase?
    \item Can we detect dense subgraphs in fixed space?
    \item Can we achieve a considerable good accuracy with some error bound in the estimation if we use probabilistic sketch datastructures?
\end{itemize}

\vspace{-8pt}
\subsection{Preliminaries}
\vspace{-6pt}
In this section, some preliminary concepts necessary for our proposed approach are discussed.

\textbf{Concept of Modularity. }
Considering, a network is divided into separate partitions, modularity proposed by \cite{newman2006modularity} represents the quality of a particular partition in the network. For an undirected graph $G$, consisting of $n$ nodes represented by a  $n $ x  $n$ real symmetric  adjacency matrix $A$, we consider having a partition of the network into $\phi$  groups. The number of edges that fall in each group is equal to $\sum_{ij}A_{ij}\delta_{\phi_i}\delta_{\phi_j}$ where $\delta_{ij}$ is the Kronecker delta function. The concept of modularity therefore is proposed as the difference between the current number and the expected number of edges placed randomly within the network. If $P_{ij}$ be the probability that nodes $i$ and $j$ are connected, then the overall modularity value of $G$ can be represented as follows:
\vspace{-6pt}
\begin{equation}
   Q = \frac{1}{2m}\sum_{i,j}(A_{ij} - \frac{k_{i}.k_{j}}{2m})\delta_{C_{i},C_{j}} 
    \label{eq:labelBasicModularity}
\end{equation}
\vspace{-10pt}

\begin{rem}[Motivation behind using Modularity]
Our main proposal is to detect based upon a dynamic setting where new edges appear, whether a new edge joins a community that has already high modularity. If that's the case, we can say it is an anomalous edge since, from the modularity function value we can see a tight modularity community will have high number of inbound edges and very low number of outbound edges. Since in such large networks, anomaly edges usually occur in sub-clusters.

\end{rem}

\noindent \textbf{Game Theoretic Framework}\label{Game theoretic framework}

 Considering a setting with $n$ players denoted as N  = ${1,2,3, \dots n}$, for each $i^{th}$ player in the game, there is a set of strategies that player takes in order to maximize its own payoff in the game. Let's say the set of such strategies be S where S = $S_1 \times S_2 \times \dots \times S_n$, combination of all the strategies of the individual players.

In a network setting, we can consider each such node as an agent that tries to maximize its own payoff, while joining a particular community formed within the network.

%\newpage

\textbf{Corresponding strategies for node $u$ include:}\label{decision}

\begin{itemize}[noitemsep,topsep=0pt,leftmargin=*]
    \item Join a particular community.    $C$ $\implies$ $C$ $\cup$ $u$
    \item  Leave a particular community $C \implies C - u$
    \item Switch communities $C_1$, $C_2$ :  $C_1 - u$ $\implies$  $C_2$ $\cup$ $u$
\end{itemize}

\textbf{Goal of such strategies. } The goal for a node is thus to choose the maximum payoff strategy, thereby maximizing the modularity of the network, leading to a stable partition, in subsequent iterations.

\begin{definition}
\textbf{Payoff function.} As per Marginal Payoff under Section 3 Definition 1. \cite{mcsweeney2012game}  The payoff function of a node $u$ for node $v$ is given as:

\end{definition}

\[q(u,v,C) = \mu(u, C \cup {v}) - \mu(u,C - {v})\] 

where $\mu(i,C)$ is  
$\sum_{j \in C}(A_{ij} -\frac{k_{i}.k_{j}}{2m})\delta_{C,C_{i}}$

The payoff function indicates how useful a particular action will be for a given agent in a game setting. In the community formation domain, payoff indicates whether joining or leaving or switching a community is profitable.
From here we use the notation for denoting the payff function as utility function interchangeably.

%\red{Remove all subsubsections make them paragraph}

\noindent \textsc{\large \bf Properties of Utility Function}

Here we consider some of the desirable properties of payoff functions in general. We follow  definitions \textbf{2} \& \textbf{3} of Section $\RomanNumeralCaps{3}$ from \cite{mcsweeney2012game} regarding the following properties.

\begin{definition}
 \textit{\textbf{Symmetric}} Let $q$ be a payoff function and $u$ and $v$ be two different nodes. Then, $q$ is symmetric iff

\[\forall_{S \in N} (q(u, v, C) = q(v, u, C))\]
\end{definition}

\begin{definition}
\textit{\textbf{Aditively-Separable}}
$q$ is additively-separable iff
$\mu(u, S)$ = $\sum_{v \in S}\delta(u, v, S)$,
The property of additively-separable states that a node u’s payoff for a community S is the sum of the marginal payoffs
over the members of S. 
\end{definition}

We now introduce here the concept of \textit{Nash equilibrium} in the context of reaching a solution to finding optimal communities using the proposed payoff function. The motivation for proposing Nash equilibrium is due to the fact that at the end of each iteration of our community detection model we want to find an optimal community partition (\textit{each node is in the best community relative to other possible alternative communities}). Since each of these nodes can take any of the strategies independently as self-rational agents, the stability of partitions formed is often decided whether the partitions have reached a Nash equilibria. We thus formally define the same.

\begin{definition}
\textbf{Nash Equilibria}
A partition $\phi$ is a Nash Equilibrium for (N, $\mu$) iff for all nodes $u$:

\[ q(u, \phi(u)) \geq q(u, C \cup {u}) \; \;  \;  \forall \;  C \in \phi, C \neq \phi(u)\]

which indicates for any partition, corresponding nodes' payoffs wont improve any further.
\end{definition}

\begin{proposition}
Every network will have at least one Nash equilibrium provided a symmetric and additively-separable payoff function is used.
\end{proposition}
This proposition is proven in \cite{bogomolnaia2002stability}.

We thus show in further discussions how our proposed payoff function would follow the utility function properties, thus achieving Nash Equilibria.

\vspace{12pt}

\noindent \textsc{\large \bf Concept of Reinforcement Learning}

\vspace{7pt}

A stochastic game which involves multiple agents each trying to maximize their own payoff can be represented as a natural extension to Markov Decision Process (MDPs) in multiple agent systems \cite{bowling2000analysis}. %\red{cite?} 
Multiagent environments are inherently non-stationary since all the agents involved in the game are free to change their behaviour. Stochastic games extend the single agent Markov decision process to include multiple agents whose actions all
impact the resulting rewards and next state.

In our case which is a multi-agent system of different nodes in the graph, we can consider the current state of the graph as $S_t$. 

So we can consider a Markov Decision Process (MDP) %\red{cite?} 
represented as a tuple ($S,A,P,R$) where $S$ is a set of states, $A$ is a state of actions, $P$ is a transition function s.t. $P: S \times  A \times S \mapsto [0,1]$ and $R$ is a reward function s.t. $R: S \times A \mapsto R$. Here, the transition function defines a probability distribution over next states as a function of the current state and the agent's action.
Here an agent/node's action will be defined by the strategies defined in \ref{decision} in our setting.

\begin{figure}[h!]
    \centering
    \includegraphics[width=0.66\linewidth]{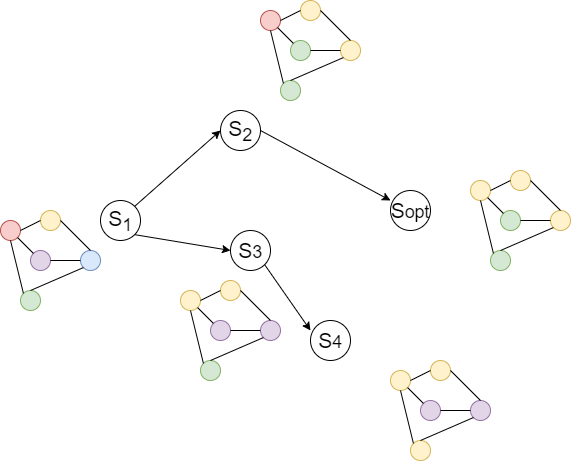}
    \caption{State Transitions to reach an optimal community partition}
    \label{fig:stateTransitionExample}
\end{figure}

The reward of each agent is to improve their payoff utility function. As a result, when a new incoming edge arrives, a state transition takes place from $s_t$ to $\Tilde{s_t}$ until at each stage we arrive at $s_t^{NE}$ where $s_t^{NE}$ denotes the optimal state arrived as per Nash Equilibrium.

In Fig \ref{fig:stateTransitionExample} we showcase how the state transition takes place. $S_1$ indicates the initial state where each node in the graph is an individual community. As the individual agents plays the game to maximize their individual rewards, it transitions to a new state. For instance, here the possible future states for the yellow node in $S_1$ is either join with the red node to transition to $S_3$ or join with the violet node to transition towards $S_2$.

\section{Proposed Methodology}\label{Proposed Methodology}

\subsection{GCD: Game theoretic based Modularity Calculation}

In this section, we first propose anomaly detection based on finding dense (of high modularity) subgraphs using the original version of game theoeretic framework \citep{newman2006modularity}, where each node in the graph tries to maximize their payoff utility in making a decision.

In our dynamic setting, for each timestamp $t_i$, we have a snapshot of the entire graph $G_{t_i}$.

%\begin{figure}[h!]
 %   \centering
  %  \includegraphics[width=0.6\linewidth]{graphsnapshots.png}
   %edges appear in $T_2$}
    %\label{fig:graphsnapshots}
%\end{figure}

\textbf{Sketch Data structure to estimate dense subgraphs or community formation. } 

Obviously, for such type of analysis we need to store information regarding the actual nodes in the previous graph snapshots. This is a challenge since in a dynamic setting there can be countably infinite possible nodes in the graph that can keep on increasing as time goes. Also, it is not possible to store all edge frequencies and degrees of the nodes in an online scenario.

We thus use sketch data-structures with constant fixed space to approximate corresponding edge frequency values as well as to consider only very relevant/important nodes in each snapshot. 

\begin{itemize}
    \item For the first part where we consider only few important relevant nodes for particular timestamp $t_i$, we use Apache's Frequent Item Sketch. \cite{ting2020data}
    \item For the second part, for storing approximate edge frequencies, we use the one most popular frequency based sketches called count-min sketch \citep{cormode2005improved}.  
\end{itemize}

Count-min sketches (CMS) are popular streaming data structures used by several online
algorithms. CMS uses multiple hash functions to map events to frequencies, but unlike a
hash table uses only sub-linear space, at the expense of overcounting some events due to collisions.
Frequency is approximated as the minimum over all hash functions. CMS, is usually
represented as a two-dimensional matrix where each row corresponds to a hash function and hashes
to the same number of buckets (columns).

As per Thm. $1$ in \cite{cormode2005improved}, the estimated frequency of an item $x$ denoted by $\hat{a}$ from a CMS has the following guarantees: $a \leq \hat{a}$ with some probability $1 - \frac{\delta}{2}$, $a$ and $\hat a$ being the actual and estimated frequency of $x$.

\begin{equation}
\begin{aligned}
\hat{a} \leq a + \epsilon V
\label{eq:CMSBounds}
\end{aligned}
\end{equation}

where the values of these error bounds can be chosen by $w = \lceil{\frac{e}{\epsilon}}\rceil$ and $d = \lceil\ln(\frac{1}{\delta})\rceil$, $w$ and $d$ being the dimensions of the CMS. $V$ indicates the total sum of all the counts stored in the data structure.

Our main goal is to find dense sub-communities that forms in the current time tick. And if for a particular incoming edge at $t$, we check if we can achieve a  partition that follows Nash Equilibrium by re-considering strategies for each of the nodes in the network based on the intervention of a new incoming edge.

For each timestamp $t$ we maintain 2 CMS:

\begin{itemize}
    \item \textbf{Edge Frequency CMS} - Here we store the edge related information for each new incoming edge and retrieval and update can be made using the corresponding edge id.
    \item \textbf{Degree CMS} - For each new incoming edge $u-v$ we hash $u$ and $v$ with the current frequency of occurrence into the the CMS. All retreival and updates can be done in this CMS via the corresponding node id and this will be storing information related to the degree of any node currently existing in the dynamic graph.
\end{itemize}

For the corresponding time tick $t$, since we know the corresponding $\langle$src, dest$\rangle$ pairs, for each node we know the corresponding information like estimated degree, estimated number of edges between any two corresponding nodes that exist in the current node, thereby detecting dense subgraphs or subcommunities.

Hence for the time tick $t$, the agent/node's payoff utility function for joining community $C$,  $\mu(i,C)$ and transitioning to state $S'$ would be as follows:

\begin{equation}
    \hat{q} = \sum_{j \in C}(\hat{A}_{ij} -\frac{\hat{k_{i}}.\hat{k_{j}}}{2\hat{m}})\delta_{C,C_{i}}
    \label{eq:approxmiate_modularity}
\end{equation}

And the overall approximate modularity is given by

\begin{equation}
    \hat{Q} = \frac{1}{2\hat{m}}\sum_{ij}(\hat{A}_{ij} -\frac{\hat{k_{i}}.\hat{k_{j}}}{2\hat{m}})\delta_{C_{i},C_{j}}
    \label{eq:total_approxmiate_modularity}
\end{equation}

Here $\hat{A}_{ij}$ denotes the estimated frequency between node $i$ and node $j$ which can be retrieved from the CMS by querying as follows, $\retrieve(i-j)$ which returns the estimated edge frequency.

$\hat{k_i}$ denotes the estimated degree of node $i$ which can be calculated via the degree CMS.

In supplementary section we showcase how the error rate in calculating the approximate modularity w.r.t actual modularity obtained is bounded by some threshold.

%guraly Scoring}

%For each incoming new edge $u-v$ consider the modularity of the community that $u$ belongs to (say, $M_u$) and correspondingly $M_v$. Then we take an average of the two such that anomaly score of edge $u-v$ is given as follows.

%\[Anom(u-v) = \frac{\lambda  M_u  + (1-\lambda) M_v}{2}\]

\subsection{TGDC: Temporal based Game-Theoretic Community Detection}

While an agent's decision from Sec. \ref{decision} is based upon immediate neighborhood information (neighbors in the closest vicinity of 1-hop distance) in the network, despite being profitable at current iteration, it might not be profitable later iterations. For instance, a node might have to choose a community it initially rejected.
We therefore consider the improvement of an agent's decision not only based on its immediate neighborhood information but also through its temporal K-hop information (in our case temporal means future possibilities of a node to join). We thus discuss the concept of community retention, i.e. if a node has a higher chance of being retained in the community it initially chose. 

\vspace{8pt}

\noindent \textsc{\large \bf Defining Temporal Relation}

\vspace{5pt}

As per Fig: \ref{fig:prefaceComm1}, initially a node takes its action based on maximizing its utility (payoff) function and decides to join a community, say $C_{1}$. However, after some iterations, it finds that the other communities its connected to via its branches are much more profitable, lets say $C_{2}$. So there is a corresponding switch action. This is a classic case where the node takes a bad decision initially based on local motivations for joining C1 instead of C2, but later switches community. 

\begin{figure}[h!]
\centering
\begin{subfigure}{.6\linewidth}
  \centering
  \includegraphics[width=0.8\linewidth]{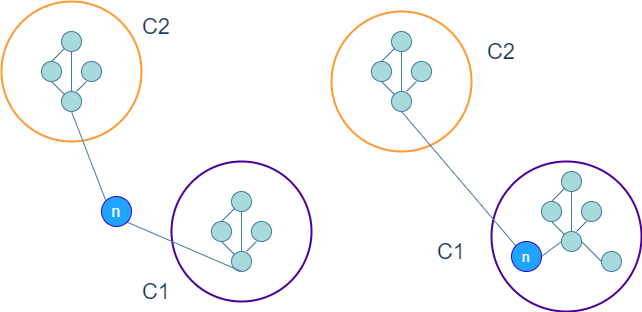}
  \caption{Node n sees joining $C_1$ is favourable in the first iteration}
  \label{fig:prefaceComm1}
\end{subfigure}%
\begin{subfigure}{.5\linewidth}
  \centering
  \includegraphics[width=.54\linewidth]{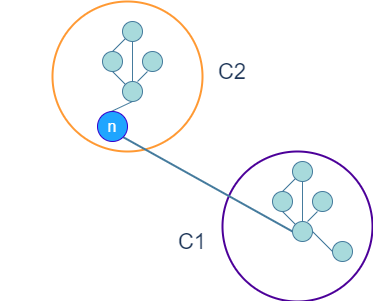}
  \caption{Node regrets in later iterations, now chooses $C_2$ which seem more profitable}
  \label{fig:sub2}
\end{subfigure}

\caption{A case of switching communities}
\label{fig:prefaceComm}
\end{figure}

\noindent \textsc{\large \bf Push and Pull Strategies}

While the general utility function indicates how profitable an agent's action will be while joining a particular community at the current instant, it does not consider whether a node will regret the current choice based on local environment and have to switch later on. We therefore propose the idea of push and pull strategies.

\noindent \textsc{\large \bf Pull strategy involves the following.
}
Suppose in \ref{fig:sub_future_dec} a node $N_i$ decides to join a community $C$, where the corresponding community has nodes which has external neighbors as 
$O_{1}$, $O_{2}$. If for a particular node belonging to the external neighbors if there are a lot of connections to them, then there is a high probability of that node getting pulled into the community in some time later, thus increasing more number of internal connections in future iterations, resulting in higher modularity of the community as a whole. So if we consider an average pull which actually indicates the expected number of nodes that will be pulled into the community, then the average number of internal connections in the future iteration will increase thereby forcing $N_i$ in Fig. \ref{fig:sub_future_dec} to be retained inside the community instead of switching.

\begin{figure}[h!]
\centering
\begin{subfigure}{.45\linewidth}
  \centering
  \includegraphics[width=.9\linewidth]{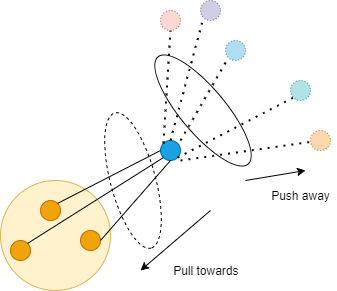}
  \caption{Push-Pull Strategy}
  \label{fig:sub1}
\end{subfigure}%
\begin{subfigure}{.55\linewidth}
  \centering
  \includegraphics[width=.87\linewidth]{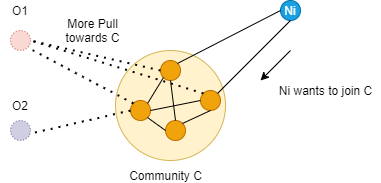}
  \caption{Ni's decision based on if $O_{1}$ would join}
  \label{fig:sub_future_dec}
\end{subfigure}
\caption{Community Retention Concept}
\end{figure}

\noindent \textsc{\large \bf Push Strategy}

However taking a decision based on this can also be detrimental if there's a case where the external degree of such a node $O_i$ is too high, indicating more external connections if considered in the community and hence more chance of escaping the community bond later. Hence nodes with high degree being entered through the pull are also prone to repulsion with the1 community later on.

\noindent \textsc{\large \bf Modified Utility Function}

The corresponding utility of the agent $i$ to join a community $C_t$ is given by follows.
\begin{equation}
\text{Payoff-Utility} = \mu(i, C_t) =  \sum_{j \in C_{t}}(A_{ij} -\frac{k_{i}.k_{j}}{2m})\delta_{C_{t},C_{i}} 
\label{eq:payoff}
\end{equation}
\begin{equation}
\text{Community-Retention} = \sum_{O_{j} \in O_{ext}}(A_{C_{t}, O_{j}} - \frac{d_{ext}^{O_{j}}}{2m})
\label{eq:commretention}
\end{equation}
\begin{equation}\label{eq4}
\lambda(\text{Payoff-Utility})+  (1-\lambda)(\text{Community-Retention})
\end{equation}

%\red{undefined notations here}

where $\lambda$ denotes the community retention rate and $\sum_{j}A_{C_{t}, O_{j}}$ includes all edge counts from community $C_t$ to external node $O_j$.

\vspace{4pt}

\textbf{Sampling Strategy.} At the initial stage, we consider each node a individual community and hence the sampling probability for each node(i.e. choosing any agent to observe the environment and take an action) follows a uniform distribution.
Since during the game, nodes who have chosen a community to join have already done so wisely taking into consideration whether they would be retained for longer term in the community, we aim at reducing the sampling space for the next iteration, by assigning a lower sampling probability to those nodes. 
Initially all the nodes have a uniform sampling probability $P$. After a node takes action based on Equation ~\ref{eq4} we assign the node a sampling probability as follows:

\begin{equation}\label{sampling_probab}
Prob(N_i) = \frac{\lambda P}{P_{N_1} + P_{N_2} + \dots + \lambda P  + P_{N_n}}
\end{equation}

The intuition is to give less chance to that node in the next iteration to be sampled, in the same proportionate amount as community retention rate. This thus guarantees that post community retention action, if the same node gets chosen again and it switches communities, its sampling probability gets factored down by $\lambda$ successively.

\vspace{5pt}

\subsection{Demonstration Example}

\begin{figure}[h!]
\centering
\begin{subfigure}{.5\linewidth}
  \centering
  \includegraphics[width=.3\linewidth]{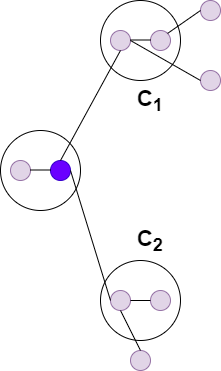}
    \caption{$C_1$: has 2 external nodes, $C_2$: has 1 external node }

  \label{fig:5sub1}
\end{subfigure}%
\begin{subfigure}{.4\linewidth}
  \centering
  \includegraphics[width=.45\linewidth]{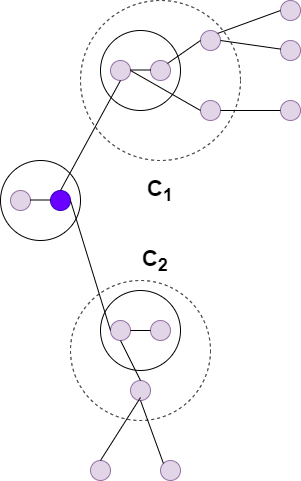}
  \caption{ Blue node decides whether to join $C_1$ or $C_2$ based on community retention }
  \label{fig:5sub2}
\end{subfigure}

\caption{Community Retention \& Sampling Strategy}
\label{fig:exampleCommunity}
\end{figure}

In the above example, we show only how the community Retention can be calculated.
The blue node chosen has an initial probability = $\frac{1}{14}$, total nodes $n$ being $14$ as per Fig \ref{fig:5sub2}.
From Fig: \ref{fig:5sub2} Total edges $m$  =  $13$, Degree $d^{ext}_{C_1}$ = $3$, Degree $d^{ext}_{C_2}$ = $2$.
Community Retention for $C_1$ = $(6  -  \frac{(3d^{ext}_{C_1} + 3d^{ext}_{C_1})}{2m}) = 5.42$ since the external nodes (not considering blue node) for $C_1$ has degree 3 and 2 respectively.
Community Retention for $C_2$ = $(4  -  \frac{(3d_{C_2}}{2m}) = 3.76$.
Considering $\lambda$ is $0.8$, then the probability of the blue node for the next iteration is 

 \[\frac{0.8 P}{P_{N_1} + P_{N_2} + \dots + 0.8  P  + P_{N_n}} = 0.0579 < P = \frac{1}{14}\]

where $P = P_{N_1} = P_{N-2} =...=P_{N_n} =$ $\frac{1}{14}$.

\vspace{6pt}

\subsection{Theoretical Proofs}\label{proofs}

\textbf{Existence of Nash Equilibrium}

%\red{existence implies what?}

\begin{proposition}
The proposed utility function in Equation. \ref{eq4} follows additive-separable and symmetric properties.
\end{proposition}

\textbf{Proof} : We show it only for the first term in Equation~\ref{eq:payoff} which can be shown for second term Equation~\ref{communityRetention} also due to the same structure.

\begin{itemize}
    \item \textbf{Symmetric Property}
    Suppose $u$ joins community which has $v$. So, symmetric property holds for the first term \textit{a} from \ref{eq:payoff}, $A_{u,v}$ = $A_{v,u}$. Also $\delta_{C_{u}, C_{v}}$ = $\delta_{C_v, C_u}$ = 1

   \item \textbf{Additively-Separable Property}

By \ref{eq:payoff}, $\delta_{C_t,C_{i}} = 1$, since node $i$ joins the $C_{t}$ community.
So $D(C, C_t)$ = $\sum_{i \in C}\mu(i, C_t)$ indicates the utility for all nodes in $C_i$ w.r.t nodes in $C_t$.
Computing this for the entire partition set $\phi$, yields, $\sum_{C \in \phi}D(C, C_t)$, which by average results in $\frac{1}{2m}\sum_{C \in \phi}\sum_{i \in C}\mu(i,S)$, i.e Q.

\end{itemize}

\noindent \textsc{\large \bf Measurement of the Quality of Community}

\textit{A. Normalized Mutual Information}\\
The measure of the mutual dependence between the communities is detected using Normalized Mutual Information (NMI) \cite{xie2013overlapping}. NMI measures the similarity between two partitions and denotes the quality of the partition. Let the two partitions or communities are $C$ and $C'$. Then NMI can be denoted as
\begin{equation}\label{eq2}
NMI(C,C')=\frac {2I(C,C')}{H(C)+H(C')} 
\end{equation}
Where $H(.)$ is the entropy function and the mutual information $I(C, C')=H(C)+H(C')-H(C,C').$ No similarity and maximum similarity of two communities manifest NMI value 0 and 1 respectively. \par

%\subsubsection{Proposed Algorithm}
Consequently, we hence propose the resulting strategy as Alg. \ref{alg:proposedalg}.

Here we consider the initial probability distribution (sampling probability) to be a uniform distribution.

Initially each node in the graph is considered as 
a community.
So $\pi = [C_1, C_2, \dots C_n]$ considering $n$ nodes in the community and $C_i$ is the community label.

\begin{algorithm}[h]
\caption{Temporal Game Theoretic Community Detection}
\begin{algorithmic}[1]
\label{alg:proposedalg}
 \STATE \textbf{Input:} {Stream of graph edges for current timestamp $t_i$} 
 \STATE Insert each corresponding node for $t_i$ in the Frequent Item Sketch $F$
 \STATE Insert incoming edge frequency in \textit{Edge-CMS} and increment the slots in \textit{Degree-CMS} for the corresponding nodes connected via edge. 
\STATE \textbf{Init:} $\pi =$ Communities formed from the previous timestamp $t_{i-1}$ 
\STATE \textbf{Init:} $\pi_{prev} = \pi$ 
\STATE \textbf{Init : } $n$ = Retrieve total nodes from $F$ with some threshold $\delta$
\STATE \textbf{Init:} Probability of each node $N_i$ (denoted by $P_{N_i}$)assigned to $1/n$
\WHILE {true}
 
       \STATE  {node $N_{t} =$  sampleNode(NodeList, $p=[p_1,..,p_{N_t}, ..., p_n]$)}   
    \STATE $ \pi = \text{CommunitySwitch}(N_{t},G, \pi)$
% \red{where is the algo?} \pctext{Algorithm 2, Had Changed the title, restored it}
    
    \IF{$N_{t}$ switches}
        \STATE $p_{N_{t}}$ $= \lambda p_{N_{t}}$
        
    \ENDIF
    
    \IF {Terminate($\pi_{prev}$, $\pi$)}
        \STATE {BREAK}
    \ENDIF
    \STATE {$\pi_{prev} = \pi$}
\ENDWHILE
\end{algorithmic}
\end{algorithm}

Using Equations \ref{eq:commretention} and \ref{eq:payoff} for Community Retention and Payoff Utility respectively and retreiving the approximate edge counts and degree counts for the current timestamp $t_i$, 

\begin{algorithm}[h!]
\caption{CommunitySwitch - Utility for a node switching communities}
\begin{algorithmic}[1]
\FOR{$each$ $neighbor$ $of$ $node$ $n$}

\STATE totalUtility = $\lambda(\text{utilityJoin})$ +
$(\text{1-}\lambda)(\text{communityRetention})$  {from Equation} \ref{eq4}
\ENDFOR

\end{algorithmic}
\end{algorithm}

\begin{algorithm}[h!]
\caption{Terminate}
\begin{algorithmic}[1]

%\STATE %\red{Explain it takes two inputs - $\pi_{prev}$ and $\pi$. Explain NMI function}

\IF{$NMI(\pi_{prev},\pi) >= \eta $}
    \STATE {\text{Terminate algorithm}}
\ENDIF
\end{algorithmic}
\end{algorithm}

\section{Experiments} \label{experiments}

In this section, we evaluate our proposed game theoretic model $\mrmd$ on first static graphs and then extend it to dynamic graphs' scenario. In further discussions, we mention how $\mrmd$ results in faster convergence w.r.t the vanilla game-theoretic model just due to the incorporation of community retention mechanism. Also we further discuss how even due to the use of probabilistic sketches, we still get quality community partitions. In the end, we showcase in large timeseries intrusion network datasets, how the dense subcommunities are captured and how they are validated against ground-truth data.

\noindent \textsc{\large \bf Comparison with other community detection algorithms.}
The Baseline Models we used are:

\begin{itemize}
    \item \textit{MMSB} \citep{airoldi2008mixed} uses dense subgraph extraction to detect overlapping communities in network graphs.
    \item \textit{CPM} \citep{palla2005uncovering} uses k-clique information to generate corresponding communities.
    \item \textit{Node2vec} \citep{grover2016node2vec} uses vertex embeddings learned via biased random walk.
    \item \textit{Fast-Unfold} \citep{blondel2008fast} is a community detection algorithm that tries to maximize the modularity using louvain heuristics.
    \item \textit{Greedy MM} Newman et al \cite{clauset2004finding} starts with each node in its own community and tries to join pairs of communities until no such community pair is left.
    \item \textit{GraphGAN} \citep{wang2018graphgan} uses adversarial training in a min-max game and combines generative and discriminative graph representation learning methods for finding communities.
    \item \textit{GCD(Vanilla Game-theoretic Community Detection)} We use \cite{mcsweeney2012game}'s basic community detection method using the node mechanism structure and aim to compare the convergence rate. This is also the case for $\lambda = 1$ in our model (i.e. without the community retention feature) and without any community merge.
\end{itemize}

\paragraph{Datasets Used. }

We use the following datasets for evaluation purpose. Here we compare first our proposed model $\mrmd$ with respect to the above baselines.

\begin{itemize}
\item \textit{Amazon} is collected by crawling Amazon website. The vertices represent products; the edges indicate the frequently co-purchase relationships; the ground-truth communities are defined by the product categories in Amazon. This graph has 3,225 vertices and 10,262 edges

\item \textit{ERDOS 992}  proposes a pajek network dataset of 6.1K nodes and 7.5K edges.

\item \textit{Facebook Food Network } Data collected about Facebook pages (November 2017). These datasets represent blue verified Facebook page networks of different categories. Nodes represent the pages and edges are mutual likes among them.

\item \textit{Retweet Network}  Nodes are twitter users and edges are retweets. These were collected from various social and political hashtags.

\item \textit{Facebook Politicians Network }  Data collected about Facebook pages (November 2017). These datasets represent blue verified Facebook page networks of different categories. Nodes represent the pages and edges are mutual likes among them.

\end{itemize}

In Table \ref{tab:table2} we report the modularity values of the overall static graphs obtained for each of the baseline models. For $\mrmd$ we consider the community retention rate to be $\lambda = 0.8$. All the reported modularity scores are considered over an average of 5 runs.

\paragraph{Experimental Setup. }
All experiments are carried out on a $1.00 GHz$ Intel Core $i5$ processor, $8 GB$ RAM, running OS $Win10$ $10.0.19042$. Under Appendix we discuss in more detail the corresponding sketch datastructure dimensions corresponding to each of the dataset.

As follows from Table \ref{tab:table2}, $\mrmd$ reflects a good quality of communities formed as compared to the baselines. In datasets like Amazon, Enron and Fb-pages-politician, $\mrmd$ shows a better modularity value.

\begin{table*}[ht!]
\centering
\begin{adjustbox}{max width=0.9\linewidth}

\begin{tabular}{| c | c | c | c | c | c| c| c| c| } 
\hline
 Network Name & TGDC(Our Model) & MMSB & CPM & Node2vec & Fast-unfold & Greedy MM & GraphGAN\\
 \hline
 Amazon & \textbf{0.77 $\pm$ 0.01} & 0.63 $\pm$ 0.013 & 0.56 $\pm$ 0.007 & 0.59  $\pm$ 0.03 & 0.67 $\pm$ 0.02 & 0.62 $\pm$ 0.04 & 0.58 $\pm$ 0.01  \\
  \hline
 Enron & \textbf{0.54 $\pm$ 0.01} & 0.418 $\pm$ 0.008 & 0.380 $\pm$ 0.011 & 0.38 $\pm$ 0.026 & 0.42 $\pm$ 0.015 & 0.47 $\pm$ 0.02 & 0.51 $\pm$ 0.02\\ 
 \hline
 ERDOS992 & 0.618 $\pm$ 0.002 & 0.604 $\pm$ 0.02 & 0.549 $\pm$ 0.064 & 0.573 $\pm$ 0.037 & 0.58 $\pm$ 0.02 & 0.571 $\pm$ 0.07 & \textbf{0.63 $\pm$ 0.01}  \\  
 \hline
 FB-Pages-Food & 0.61 $\pm$ 0.007 & 0.563 $\pm$ 0.014 & 0.503 $\pm$ 0.03 &  0.518 $\pm$ 0.011 & 0.55 $\pm$ 0.01 & 0.53 $\pm$ 0.03 & \textbf{0.61  $\pm$ 0.04}\\
 \hline
 Retweet & 0.51 $\pm$ 0.01 & 0.55 $\pm$ 0.023 & 0.53 $\pm$ 0.036 & \textbf{0.57  $\pm$ 0.046} & 0.48 $\pm$ 0.02 & 0.52 $\pm$ 0.03 & 0.57 $\pm$ 0.03\\
 \hline
 FB-Pages-Politician & \textbf{0.71 $\pm$ 0.02} & 0.55 $\pm$ 0.023 & 0.53 $\pm$ 0.036 & 0.57  $\pm$ 0.046 & 0.64 $\pm$ 0.01 & 0.61 $\pm$ 0.003 & 0.704 $\pm$ 0.04\\

 \hline
 
\end{tabular}
\end{adjustbox}

\caption{Modularity Comparison for ($\lambda$ = 0.8)}
\label{tab:table2}

\end{table*}

\paragraph{Fast Convergence of TGDC. }
We showcase through Fig \ref{fig:fasterConvergencePlot}, how $\mrmd$ results in faster convergence for finding the overall modularity, by minimizing false regrettable community switches based on reward from local environment.

\begin{figure}[ht!]
\centering
\begin{subfigure}{.8\linewidth}
  \centering
  \includegraphics[width=0.6\linewidth]{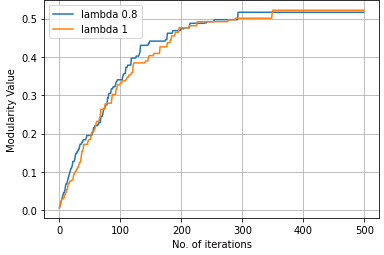}
  \caption{\textbf{FB pages Food Dataset}}
  \label{fig:fastsub2}
\end{subfigure}

\begin{subfigure}{.8\linewidth}
  \centering
  \includegraphics[width=0.6\linewidth]{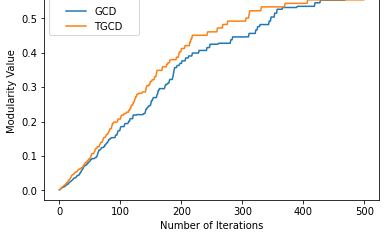}
  \caption{\textbf{Enron Dataset} - From iteration no. 100 to 400, TGDC($\lambda$ = 0.8) has higher modularity values compared to GCD($\lambda$ = 1)}
  \label{fig:fastsub3}
\end{subfigure}

\begin{subfigure}{.8\linewidth}
  \centering
  \includegraphics[width=0.6\linewidth]{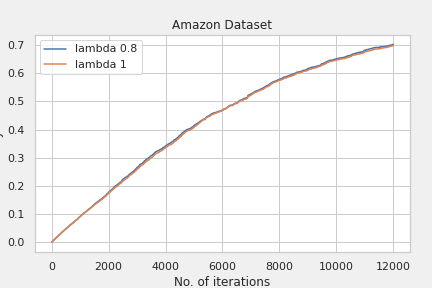}
  \caption{\textbf{Amazon Dataset} - From iteration no. 2000 to 4000, TGDC($\lambda$ = 0.8) has higher modularity values compared to GCD($\lambda$ = 1)}. Similarly for later phases.
  \label{fig:fastsub3}
\end{subfigure}
\caption{Faster Convergence for TGDC model}
\label{fig:fasterConvergencePlot}
\end{figure}

\subsection{Finding Dense Communities - Indication of Anomalies} \label{anomaly indication}

%\paragraph{Datasets}

We evaluate the identification of dense communities (i.e. individual communities with a very high modularity value greater than some threshold) in large scale dynamic networks. 
We showase the same for the following intrusion detection based timeseries datasets : \emph{DARPA} \cite{lippmann1999results} and \emph{ISCX-IDS2012} are popular datasets for graph anomaly detection.

\paragraph{Experimental Setup.}
For our experiments, for each timestamp $t_i$ we consider a static snapshot of the graph, based upon the incoming edges. We keep information about the nodes for each timestamp in a Frequent Item Sketch of \textit{maxMapSize} = 4096 (with corresponding threshold set to 0.3), Edge Count Min Sketch of dimensions $w=719,d=2$ and Degree Count Min Sketch of dimensions $w=719,d=2$.

\begin{figure}[h!]
\centering
\begin{subfigure}{0.95\linewidth}
 \centering
\includegraphics[width=0.95\linewidth]{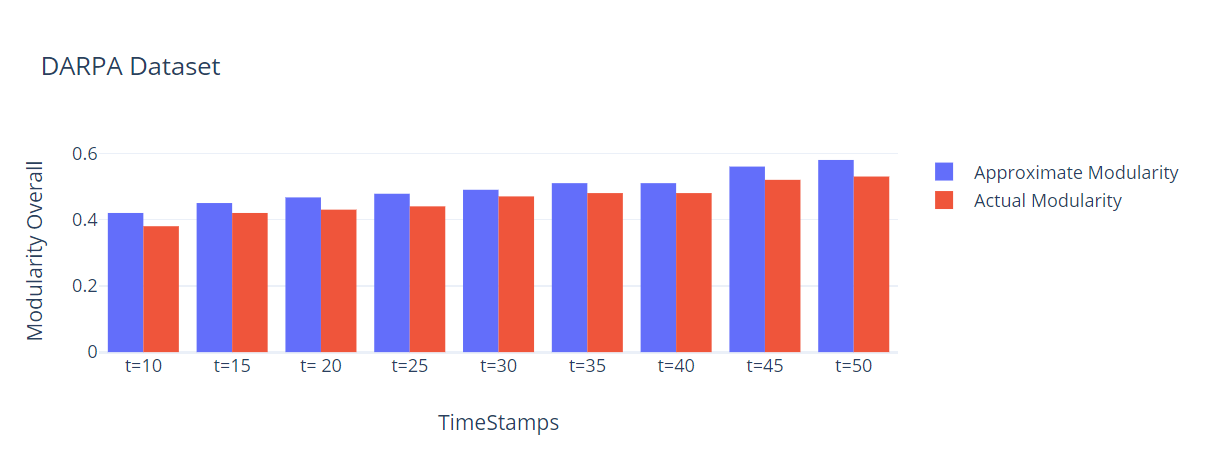}
\label{fig:DARPA_TIMESERIES}      
\end{subfigure}
\begin{subfigure}{0.95\linewidth}
 \centering
\includegraphics[width=0.95\linewidth]{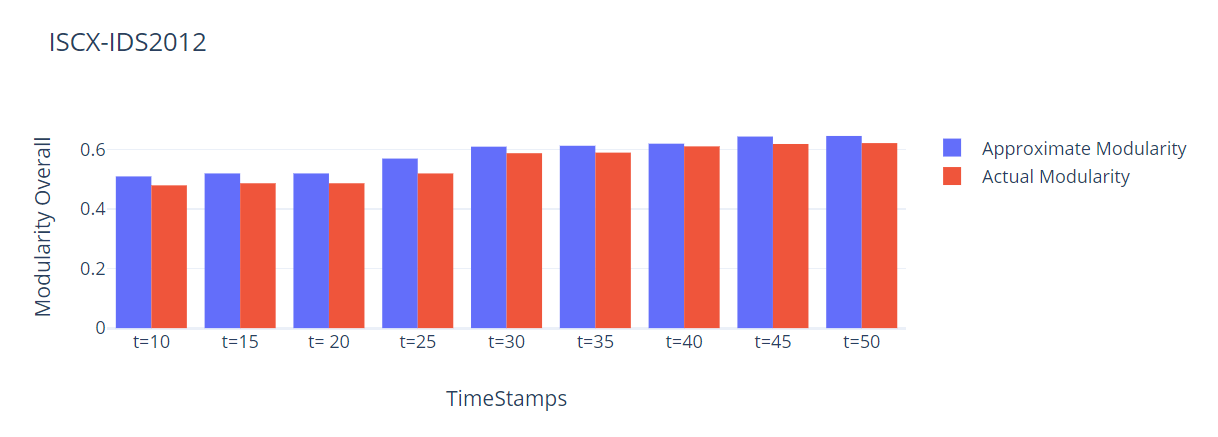}
\label{fig:IC_TIMESERIES}      
\end{subfigure}
\caption{Approximate Modularity vs Actual Modularity Values}
\label{fig:ApproximateModularity}
\end{figure}

As we see from Fig: \ref{fig:ApproximateModularity} over the first 50 timestamps, the approximate modularity achieved using the sketch datastructure is closer to the actual modularity obtained when no using constant space datastructures. For \emph{Darpa} the average error rate = \textit{0.08\%} while for \emph{ISCX-IDS2012} dataset the average error rate = \textit{0.17\%}

\paragraph{Ground Truth Comparison.}
For each of the communities we retrieved per dataset for a particular timestamp, we computed the top-K (some threshold) communities scoring by high value of modularity obtained using $\mrmd$. And for each of these $K$ communities per timestamp we observe the percentage of edges involved in the high modular communities which are actually anomaly as per our ground truth dataset. Setting $K = 10$ for the entire DARPA dataset, we observe an average accuracy of $94\%$ i.e. edges that are actually anomalous and that exists in the top 10 highest modular communities on an average . Similarly for $ISCX-IDS2012$, we observed accuracy of $91\%$. This shows how using $\mrmd$ we were able to approximate dense subcommunities which involved a high number of anomalous edges.

\vspace{-10pt}
\section{Conclusion and Future Work}
\vspace{-10pt}
We show cased how through our proposed game-theoretic model and probabilistic sketch datastructures we were able to retrieve the dense subcommunities for each timestamp with sufficient guarantees. In addition we showcased how our model performs well with constant space requirements w.r.t to traditional community detection algorithms and as a result could be extended to a dynamic setting. As one of the challenges in such a multi-agent system is that time complexity is a bit high due to agents playing towards rewards maximization. Hence, in future we aim to look at improving its performance by proposing a distributed approach where some of the agents within the environment can take action parallely.

%\newpage

\bibliographystyle{plainnat}
\bibliography{sample}

\newpage
\onecolumn
\allowdisplaybreaks

\par\noindent\rule{\textwidth}{1pt}
\begin{center}
\large\textbf{Supplementary Material: 
\papertitle}
\end{center}
\par\noindent\rule{\textwidth}{0.4pt}

\begin{proposition}\label{suppl : proposition_proof}
The error estimate of the approximate modularity value $\hat{Q}$ obtained via Equation \ref{eq:total_approxmiate_modularity} w.r.t the actual modularity $Q$ is bounded by some threshold $\Tilde{\epsilon}V_u$ %\red{some threshold}. -- \red{not a precise statement. what is "some?", write the exact statements what are we going to prove}
\end{proposition}

\begin{proof}

The approximate modularity obtained from the sketch data-structure can be shown as follows:

\begin{equation}
    \hat{Q} = \frac{1}{2\hat{m}}\sum_{i,j}(\hat{A_{ij}} - \frac{\hat{k_i}.\hat{k_j}}{2\hat{m}})\delta_{C_{i},C_{j}}
\label{Approximate: Modularity}
\end{equation}
    where $i,j$ indicates the corresponding nodes in the graph $G$.

Considering our \textit{Edge-CMS} which will store the edge wise frequencies and \textit{Degree-CMS} which will store node wise degree frequencies.

As per CMS guarantees, we know

\begin{equation}
    \hat{A_{ij}} \leq A_{ij} + \epsilon_{e}V_{e}
\label{Approximate: Edge}
\end{equation}

where $\epsilon_{e}$ is determined by the configuration parameters of the \textit{Edge-CMS} and $V_{e}$ indicates the total counts of all the frequencies in \textit{Edge-CMS}

From Equation. \ref{Approximate: Modularity}, $\hat{m}$ indicates the approximate total number of edges found in the current Graph $G$. Hence $\hat{m} = V_{e}$ in our case.

We rewrite Equation. \ref{Approximate: Edge} as follows
\begin{equation}
        \hat{A_{ij}} = A_{ij} + \epsilon_{e}V_{e} - \phi_{e_{ij}}
\label{Approximate : Edge equality}
\end{equation}
Given $\phi_{e_ij} \geq 0$

Furthering Equation. $\ref{Approximate: Edge}$ for our approximate degree values, we show:

\begin{equation}
    \hat{k_i} = k_i + \epsilon_{d}V_{d} - \phi_{d_i}
\label{Approximate : Node equality}
\end{equation}

Similarly, $\epsilon_{d}$ is determined by the configuration parameters of the \textit{Degree-CMS} and $V_{d}$ indicates the total counts of all the frequencies in \textit{Degree-CMS}

Combining Equation. \ref{Approximate: Modularity}, \ref{Approximate : Edge equality}, \ref{Approximate : Node equality} we get the following :

\[\hat{Q} = \frac{1}{2V_e}\sum_{i,j}({A_{ij} + \epsilon_{e}V_e - \phi_{e_{ij}}} - \frac{(k_i + \epsilon_dV_d - \phi_{d_i}).(k_j + \epsilon_dV_d - \phi_{d_j})}{2V_e})\delta_{C_{i},C_{j}}
\]

\[
= \frac{1}{2V_e}\sum_{i,j}({A_{ij}} + \epsilon_{e}V_e - \phi_{e_{ij}}  - \frac{k_i.k_j}{2V_e} - \frac{k_i.(\epsilon_dV_d - \phi_{d_j})}{2V_e} - \frac{k_j.(\epsilon_dV_d - \phi_{d_i})}{2V_e}
- \frac{(\epsilon_dV_d - \phi_{d_i}).(\epsilon_dV_d - \phi_{d_j})}{2V_e}
)
\]

\begin{equation}\label{eq:final approximate}   
\begin{split}
= \frac{1}{2V_e}\sum_{i,j}(A_{ij} -  \frac{k_i.k_j}{2V_e} ) + \frac{1}{2V_e}\sum_{ij}(\epsilon_eV_e - \phi_{e_{ij}})
- \frac{1}{2V_e}\sum_{i,j}\frac{k_i.(\epsilon_dV_d - \phi_{d_j})}{2V_e} -  \\ 
\frac{1}{2V_e}\sum_{i,j}\frac{k_j.(\epsilon_dV_d - \phi_{d_i})}{2V_e} -
\frac{1}{2V_e}\sum_{i.j}\frac{(\epsilon_dV_d - \phi_{d_i}).(\epsilon_dV_d - \phi_{d_j})}{2V_e}     
\end{split} 
\end{equation}

Since we know that in a graph, sum of degrees of all nodes is equal to twice the number of edges, 

\[ 
\sum_{i}k_i = \sum_{j}k_j = 2V_e
\]

the Equation. \ref{eq:final approximate} can be written as follows

\begin{equation}\label{eq : final expanded approx Q} 
\begin{split}
     = \frac{1}{2V_e}\sum_{i,j}(A_{ij} -  \frac{k_i.k_j}{2V_e} ) + \frac{1}{2V_e}\sum_{ij}(\epsilon_eV_e - \phi_{e_{ij}}) -  \frac{1}{2V_e}\sum_{j}(\epsilon_dV_d - \phi_{d_j}) - \frac{1}{2V_e}\sum_{i}(\epsilon_dV_d - \phi_{d_i}) - \\ \frac{1}{2V_e}\sum_{i.j}\frac{(\epsilon_dV_d - \phi_{d_i}).(\epsilon_dV_d - \phi_{d_j})}{2V_e}
\end{split}
\end{equation}

Considering the following notations for the first 2 variables of the above equation,

\begin{equation}
\Tilde{Q} = \frac{1}{2V_e}\sum_{i,j}(A_{ij} -  \frac{k_i.k_j}{2V_e} )
\label{replacement : Approx Q}
\end{equation}

\begin{equation}
    \frac{1}{2V_e}\sum_{ij}(\epsilon_eV_e - \phi_{e_{ij}})  = \sum_{ij}(\frac{\epsilon_e}{2} - \frac{\phi_{e_{ij}}}{2V_e})
\label{eq:approximate epsilon}    
\end{equation}

From Equation. \ref{eq:approximate epsilon}  $\frac{\phi_{e_{ij}}}{2V_e} \in \mathbb{R^{+}}$ denotes the corresponding \textit{EdgeCMS} threshold for the edge $e_{ij}$ as per Equation. \ref{Approximate : Edge equality} w.r.t the total number of edges in $G$. Si So by denoting $\frac{\phi_{e_{ij}}}{2V_e}$ as $\phi_{avg}$ Equation. \ref{eq:approximate epsilon} can be rewritten as follows ($V_u$ denotes the total number of unique edges in the graph):

\begin{equation}
    \sum_{ij}(\frac{\epsilon_e}{2} - \phi_{avg}) = V_u(\frac{\epsilon_e}{2} - \phi_{avg}) 
    = \Tilde{\epsilon}V_u
\label{eq:final Epsilon value}
\end{equation}

Thus Equation \ref{eq : final expanded approx Q} can be rewritten as follows:

\begin{equation}
\hat{Q} = 
\Tilde{Q} + \Tilde{\epsilon}V_u - (\frac{1}{2V_e}\sum_{j}(\epsilon_dV_d - \phi_{d_j}) + \frac{1}{2V_e}\sum_{i}(\epsilon_dV_d - \phi_{d_i}) +  \frac{1}{2V_e}\sum_{i.j}\frac{(\epsilon_dV_d - \phi_{d_i}).(\epsilon_dV_d - \phi_{d_j})}{2V_e} )
\label{eq : finalAPRFRGRG}
\end{equation}
The second part of the above equation can be further condensed as follows:

\[ (\frac{1}{2V_e}\sum_{j}(\epsilon_dV_d - \phi_{d_j}) + \frac{1}{2V_e}\sum_{i}(\epsilon_dV_d - \phi_{d_i}) +  \frac{1}{2V_e}\sum_{i.j}\frac{(\epsilon_dV_d - \phi_{d_i}).(\epsilon_dV_d - \phi_{d_j})}{2V_e} )  
\]

\[
= \sum_{j}(\frac{\epsilon_dV_d}{2V_e} - \frac{\phi_{d_j}}{2V_e}) + \sum_{i}(\frac{\epsilon_dV_d}{2V_e} - \frac{\phi_{d_i}}{2V_e}) + 
\sum_{i.j}(\frac{\epsilon_dV_d}{2V_e} - \frac{\phi_{d_i}}{2V_e}).(\frac{\epsilon_dV_d}{2V_e} - \frac{\phi_{d_j}}{2V_e})
\]

Since $v_d =  2V_e$ (sum of degrees of nodes in a graph is equal to twice the sum of the number of edges), the above equation can be further simplified as :

\[
= \sum_{j}(\epsilon_d - \frac{\phi_{d_j}}{2V_e}) + \sum_{i}(\epsilon_d - \frac{\phi_{d_i}}{2V_e}) + 
\sum_{i.j}(\epsilon_d - \frac{\phi_{d_i}}{2V_e}).(\epsilon_d - \frac{\phi_{d_j}}{2V_e})
\]

which can be represented by some constant $\phi(n,V_e)$ based upon the parameters $n$ and $V_e$, where $n$ indicates the number of nodes in the graph.

Hence combining Equations \ref{replacement : Approx Q} , \ref{eq:approximate epsilon}, \ref{eq:final Epsilon value}, \ref{eq : finalAPRFRGRG} , we can simplify the following equation as

\[\hat{Q} = 
\Tilde{Q} + \Tilde{\epsilon}V_u - \phi(n, V_e)    
\]

where appropriate parameters for the sketch datastructures are chosen such that $\Tilde{\epsilon}V_u - \phi(n, V_e)  \geq 0$

\[\hat{Q} + \phi(n, V_e)  = 
\Tilde{Q} + \Tilde{\epsilon}V_u  
\]
\[
\hat{Q} \leq \hat{Q} + \phi(n, V_e)  = 
\Tilde{Q} + \Tilde{\epsilon}V_u  
\]

Thus the final inequality we achieve is the following

\begin{equation}
    \hat{Q} \leq 
\Tilde{Q} + \Tilde{\epsilon}V_u  
\end{equation}
\label{finalap}

Since, $V_e = \hat{m} \geq m$ , where $m$ is the total number of edges as per Equation \ref{Approximate: Edge}, hence it can be proved that $\Tilde{Q} \leq Q$ where $Q$ is given by

\[
Q =  \frac{1}{2m}\sum_{i,j}(A_{ij} -  \frac{k_i.k_j}{2m} )
\]

Hence Equation \ref{finalap} can be showcased as

\begin{equation}
    \hat{Q} \leq Q + \Tilde{\epsilon}V_u  
\end{equation}
 
Thus proving our proposition that the approximate modularity value obtained is bounded by some error threshold.

\end{proof}

\appendix

\end{document}